\documentclass[conference]{IEEEtran}

\usepackage{times,latexsym,amsfonts,amsmath,amssymb,mathrsfs,verbatim,cite}
\usepackage{epsfig,epsf}
\usepackage{graphicx}
\usepackage[dvips]{color}
\usepackage{txfonts}

\def\zZ{{\mathbb Z}}
\def\rR{{\mathbb R}}
\def\eE{{\mathbb E}}
\def\pP{{\mathbb P}}

\def\QED{\mbox{\rule[0pt]{1.5ex}{1.5ex}}}

\def\@begintheorem#1#2{\tmpitemindent\itemindent\topsep 0pt\rm\trivlist
    \item[\hskip \labelsep{\indent\it #1\ #2:}]\itemindent\tmpitemindent}
\def\@opargbegintheorem#1#2#3{\tmpitemindent\itemindent\topsep 0pt\rm \trivlist
    \item[\hskip\labelsep{\indent\it #1\ #2\
    \rm(#3):}]\itemindent\tmpitemindent}
\def\@endtheorem{\endtrivlist\unskip}

\newtheorem{theorem}{Theorem}
\newtheorem{definition}{Definition}

\newtheorem{proposition}{Proposition}
\newtheorem{lemma}{Lemma}

\newtheorem{remark}{Remark}
\renewcommand{\theequation}{\arabic{section}.\arabic{equation}}

\newcommand{\supp}{\operatorname{supp}}
\begin{document}

\title{Interaction Strictly Improves the Wyner-Ziv Rate-distortion
Function$^{\text{\small 1}}$}



\author{\IEEEauthorblockN{Nan Ma}
\IEEEauthorblockA{ECE Dept, Boston University\\
Boston, MA 02215 \\ {\tt nanma@bu.edu}}
\and
\IEEEauthorblockN{Prakash Ishwar}
\IEEEauthorblockA{ECE Dept, Boston University\\
 Boston, MA 02215 \\ {\tt pi@bu.edu}}
}

\maketitle

\begin{abstract}
In 1985 Kaspi provided a single-letter characterization of the
sum-rate-distortion function for a two-way lossy source coding problem
in which two terminals send multiple messages back and forth with the
goal of reproducing each other's sources. Yet, the question remained
whether more messages can strictly improve the sum-rate-distortion
function. Viewing the sum-rate as a functional of the distortions and
the joint source distribution and leveraging its convex-geometric
properties, we construct an example which shows that two messages can
strictly improve the one-message (Wyner-Ziv) rate-distortion
function. The example also shows that the ratio of the one-message
rate to the two-message sum-rate can be arbitrarily large and
simultaneously the ratio of the backward rate to the forward rate in
the two-message sum-rate can be arbitrarily small.
\end{abstract}

\section{Introduction}
\addtocounter{footnote}{+1} \footnotetext{This material is based upon
work supported by the US National Science Foundation (NSF) under award
(CAREER) CCF--0546598 and CCF--0915389. Any opinions, findings, and
conclusions or recommendations expressed in this material are those of
the authors and do not necessarily reflect the views of the NSF.}

Consider the following two-way lossy source coding problem studied in
\cite{Kaspi1985}. Let $(X(1),Y(1)), \ldots, (X(n), Y(n))$ be $n$ iid
samples of a two-component discrete memoryless stationary source with
joint pmf $p_{XY}(x,y)$, $(x, y) \in \mathcal X \times \mathcal Y$,
$|\mathcal X \times \mathcal Y|< \infty$. Terminal A observes $\mathbf
X:=(X(1),\ldots,X(n))$ and terminal B observes $\mathbf
Y:=(Y(1),\ldots,Y(n))$. Terminal B is required to produce
$\widehat{\mathbf X}:=(\widehat X(1),\ldots,\widehat X(n))\in
\widehat{\mathcal X}^n$, where $\widehat{\mathcal X}$ is a
reproduction alphabet with $|\widehat{\mathcal X}|<\infty$, such that
the expected distortion $\eE[d^{(n)}(\mathbf X,\widehat{\mathbf X})]$
does not exceed a desired level, where
\[d^{(n)}(\mathbf x,\hat{\mathbf x}):= \frac{1}{n}\sum_{i=1}^n
d(x(i),\hat{x}(i)),\] and $d: \mathcal X \times \widehat{\mathcal
X}\rightarrow \rR^+\bigcup \{\infty\}$ is a per-sample (single-letter)
distortion function. Terminal A is likewise required to reproduce the
source observed at terminal B within some distortion level with
respect to another (possibly different) distortion function. To
achieve this objective, the terminals are allowed to send a certain
number of messages back and forth where each message sent from a
terminal at any time only depends on the information available at the
terminal up to that time. In \cite{Kaspi1985}, Kaspi provided a
single-letter characterization of the sum-rate-distortion function for
any finite number of messages.  Yet, whether more messages can
strictly improve the sum-rate-distortion function was left
unresolved. If the goal is to reproduce both sources {\em losslessly}
at each terminal (zero distortion) then there is no advantage in using
multiple messages; two messages are sufficient and the minimum
sum-rate cannot be reduced by using more than two
messages.\footnote{If only one of the sources is required to be
losslessly reproduced at the other terminal then one message is
sufficient and the minimum sum-rate cannot be improved by using more
than one message. However, if $\bf X$ and $\bf Y$ are nonergodic,
two-way interactive coding can be strictly better than one-way
non-interactive coding\cite{Dakehe}.} If, however, the goal is changed
to losslessly {\em compute functions} of sources at each terminal,
then multiple messages can decrease the minimum sum-rate by an
arbitrarily large factor \cite{OrlitskyRoche,ISIT08}. Therefore, the
key unresolved question pertains to {\em lossy source reproduction}:
can multiple messages strictly decrease the minimum sum-rate for a
given (nonzero) distortion? This question is unresolved even when only
one source needs to be reproduced with nonzero distortion.

In this paper, we construct the first example which shows that two
messages can strictly improve the one-message (Wyner-Ziv)
rate-distortion function. The example also shows that the ratio of the
one-message rate to the two-message sum-rate can be arbitrarily large
and simultaneously the ratio of the backward rate to the forward rate
in the two-message sum-rate can be arbitrarily small.  The key idea
which enables the construction of this example is that the sum-rate is
a {\em functional} of the distortion and the joint source distribution
which has certain convex-geometric properties.

\section{Problem setup and related prior results}
\label{sec:problem}
\subsection{One-message Wyner-Ziv rate-distortion
function}\label{subsec:onemessage}
\begin{definition}\label{def:onemessagecode}
A one-message distributed source code with parameters $(n,|{\mathcal
M}|)$ is the tuple $(e^{(n)},g^{(n)})$ consisting of an encoding
function $e^{(n)}: \mathcal X^n \rightarrow \mathcal M$ and a decoding
function $g^{(n)}: \mathcal Y^n \times \mathcal M \rightarrow
\widehat{\mathcal X}^n$. The output of $g^{(n)}$, denoted by
$\widehat{\mathbf X}$, is called the reproduction and $(1/n) \log_2
|{\mathcal M}|$ is called the block-coding rate (in bits per sample).
\end{definition}

\begin{definition}\label{def:onemessagerate}
A tuple $(R,D)$ is admissible for one-message distributed source
coding if, $\forall \epsilon > 0$, $\exists~ \bar n(\epsilon)$ such
that $\forall n> \bar n(\epsilon)$, there exists a one-message
distributed source code with parameters $(n,|{\mathcal M}|)$
satisfying $\frac{1}{n}\log_2 |{\mathcal M}| \leq R + \epsilon,$ and
$\eE[d^{(n)}(\mathbf X,\widehat{\mathbf X})]\leq D + \epsilon.$
\end{definition}

The set of all admissible $(R,D)$ tuples in
Definition~\ref{def:onemessagerate} is a closed subset of $\rR^2$. For
any $D\in \rR$, the minimum value of $R$ such that $(R,D)$ is
admissible is the one-message Wyner-Ziv rate-distortion
function\cite{WynerZiv1976} and will be denoted by $R_{sum,1}(D)$. The
following single-letter characterization of $R_{sum,1}(D)$ was
established in \cite{WynerZiv1976}:
\begin{equation}\label{eqn:onemsgrate}
R_{sum,1}(D)=\min_{p_{U|X},g:~ \eE[d(X,g(U,Y))]\leq D} I(X;U|Y),
\end{equation}
where $U\in \mathcal U$ is an auxiliary random variable such that $U -
X - Y$ is a Markov chain and $|\mathcal U|\leq |\mathcal X|+1$, and
$g: \mathcal U\times \mathcal Y \rightarrow \widehat{\mathcal X}$ is a
deterministic single-letter decoding function.


\subsection{Two-message sum-rate-distortion
function}\label{subsec:twomessage}
\begin{definition}\label{def:twomessagecode}
A two-message distributed source code with parameters $(n,|{\mathcal
M_1}|,|{\mathcal M_2}|)$ is the tuple $(e_1^{(n)},e_2^{(n)},g^{(n)})$
consisting of encoding functions $e_1^{(n)}: \mathcal Y^n \rightarrow
\mathcal M_1$, $e_2^{(n)}: \mathcal X^n \times \mathcal M_1
\rightarrow \mathcal M_2$ and a decoding function $g^{(n)}: \mathcal
Y^n \times \mathcal M_1\times \mathcal M_2 \rightarrow
\widehat{\mathcal X}^n$. The output of $g^{(n)}$, denoted by
$\widehat{\mathbf X}$, is called the reproduction and for $i=1,2$,
$(1/n) \log_2 |{\mathcal M_i}|$ is called the $i$-th block-coding
rate.
\end{definition}

\begin{definition}\label{def:twomessagerate}
A tuple $(R_1,R_2,D)$ is admissible for two-message distributed source
coding if, $\forall \epsilon > 0$, $\exists~ \bar n(\epsilon)$ such
that $\forall n> \bar n(\epsilon)$, there exists a two-message
distributed source code with parameters $(n,|{\mathcal
M_1}|,|{\mathcal M_2}|)$ satisfying $\frac{1}{n}\log_2 |{\mathcal
M_i}| \leq R_i + \epsilon,$ for $i=1,2$, and $\eE[d^{(n)}(\mathbf
X,\widehat{\mathbf X})]\leq D + \epsilon.$
\end{definition}

The rate-distortion region, denoted by $\mathcal{RD}$, is defined as
the set of all admissible $(R_1,R_2,D)$ tuples and is a closed subset
of $\rR^3$. For any $D\in \rR$, the minimum value of $(R_1+R_2)$ such
that $(R_1,R_2,D)\in \mathcal{RD}$ is the two-message
sum-rate-distortion function and will be denoted by
$R_{sum,2}(D)$. The following single-letter characterization of
$\mathcal{RD}$ was established in \cite{Kaspi1985}:
\begin{eqnarray}
{\mathcal {RD}} = &\{&\!\!\!~(R_1,R_2,D)~|~\exists \
p_{V_1|Y},p_{V_2|XV_1},g,
s.t.\nonumber\\
&&R_1 \geq I(Y;V_1|X),\nonumber\\
&&R_2 \geq I(X;V_2|Y,V_1), \nonumber\\
&& \eE[d(X,g(V_1,V_2,Y))]\leq D ~\},\label{eqn:rateregion}
\end{eqnarray}
where $V_1\in \mathcal V_1$ and $V_2\in \mathcal V_2$ are auxiliary
random variables with bounded alphabets,\footnote{Bounds for the
cardinalities of $\mathcal V_1$ and $\mathcal V_2$ can be found in
\cite{Kaspi1985}.} such that the Markov chains $V_1 - Y - X$ and $V_2
- (X,V_1) - Y$ hold, and $g: \mathcal V_1 \times \mathcal V_2 \times
\mathcal Y\rightarrow \widehat{\mathcal X}$ is a deterministic
single-letter decoding function.
From (\ref{eqn:rateregion}), it follows
that
\begin{equation}\label{eqn:twomsgrate}
R_{sum,2}(D)=\min_{\scriptstyle p_{V_1|Y},p_{V_2|XV_1},g:\atop
\scriptstyle \eE[d(X,g(V_1,V_2,Y))]\leq D} \{I(Y;V_1|X)+I(X;V_2|Y,
V_1)\}.
\end{equation}

Since a one-message code is a special case of a two-message code with
$|\mathcal M_1|=1$, the inequality $R_{sum,2}(D)\leq R_{sum,1}(D)$
holds for all $D\in \rR$. Even though the single-letter
characterizations of $R_{sum,1}(D)$ and $R_{sum,2}(D)$ are known, it
has proved difficult to demonstrate the existence of $p_{XY}$, $d$,
and $D$ such that $R_{sum,2}(D) < R_{sum,1}(D)$. In the distributed
source coding literature, to the best of our knowledge, there is
neither an explicit example which shows that $R_{sum,2}(D) <
R_{sum,1}(D)$ nor an implicit proof that such an example must exist
nor a proof that there is no such example. In this paper we will
construct an explicit example for which $R_{sum,2}(D) < R_{sum,1}(D)$.

In \cite{Allerton09,Allerton09arXiv}, for a general $t\in \zZ^+$, we
established a connection between the $t$-message sum-rate-distortion
function and the $(t-1)$-message sum-rate-distortion function using
the rate reduction functional defined in the next subsection. This
connection and the properties of the rate reduction functional
allows one to compare $R_{sum,2}(D)$ and $R_{sum,1}(D)$ without
having to explicitly solve the optimization problem in
(\ref{eqn:twomsgrate}).

\subsection{Key tool: rate reduction functionals}\label{subsec:ratereduction}
Generally speaking, for $i=1,2$, $R_{sum,i}$ depends on
$(p_{XY},d,D)$. As in \cite{Allerton09,Allerton09arXiv}, we fix $d$
and view $R_{sum,i}$ as a functional of $(p_{XY},D)$. The sum-rate
needed to reproduce only terminal A's source at terminal B with
nonzero distortion can only be smaller than the sum-rate needed to
losslessly reproduce both sources at both terminals which is equal to
$H(X|Y) + H(Y|X)$. The reduction in the rate for lossy source
reproduction in comparison to lossless source reproduction of both
sources at both terminals is the rate-reduction functional.
Specifically, the rate reduction functionals \cite{Allerton09arXiv}
are defined as follows.  For $i=1,2$,
\begin{equation}
\rho_i(p_{XY},D) := H(X|Y)+H(Y|X)-R_{sum,i}(p_{XY},D).\label{eqn:rhoi}
\end{equation}
%
%
Since $R_{sum,1} \geq R_{sum,2}$ and $\rho_1 \leq \rho_2$ always hold,
$R_{sum,1}(p_{XY},D) > R_{sum,2}(p_{XY},D)$ if, and only if,
$\rho_1(p_{XY},D) < \rho_2(p_{XY},D)$, i.e., if, and only if,
$\rho_1(p_{XY},D) \neq \rho_2(p_{XY},D)$. The following key lemma
provides a means for testing whether or not $\rho_1 = \rho_2$ without
ever having to evaluate or work with $\rho_2$, i.e., without
explicitly constructing auxiliary variables $V_1, V_2$ and the
decoding function $g$ in (\ref{eqn:twomsgrate}).

\begin{lemma}\label{lem:rho1rho2}
The following two conditions are equivalent: (1) For all $p_{XY}$
and $D$, $\rho_1(p_{XY}, D)=\rho_2(p_{XY}, D)$. (2) For all
$p_{X|Y}$, $\rho_1(p_{X|Y}p_Y, D)$ is concave with respect to
$(p_{Y},D)$.
\end{lemma}

In simple terms, $\rho_1 = \rho_2$ if, and only if, $\rho_1$ is
concave under $Y$-marginal and distortion perturbations. The proof of
Lemma~\ref{lem:rho1rho2} is along the lines of the proof of part (i)
of Theorem~2 in \cite{Allerton09arXiv} and is omitted. In fact, it can
be proved that if for some $t \in \zZ^+$, the $t$-message
rate-reduction functional is identically equal to the $(t+1)$-message
rate reduction functional, i.e., $\rho_t = \rho_{t+1}$, then $\rho_t =
\rho_\infty$, the infinite-message rate-reduction functional.
As discussed in \cite[Remark~6]{Allerton09arXiv},
Lemma~\ref{lem:rho1rho2} does not hold if all the rate reduction
functionals are replaced by the sum-rate-distortion functionals.
Therefore the rate reduction functional is the key to the connection
between a one-message distributed source coding scheme and a
two-message distributed source coding scheme.

The remainder of this paper is organized as follows. In
Theorem~\ref{thm:improvement}, we will use Lemma~\ref{lem:rho1rho2} to
show that there exist $p_{XY}, d$, and $D$ for which
$R_{sum,1}(p_{XY},D) > R_{sum,2}(p_{XY},D)$. We will do this by (i)
choosing $p_{X|Y}$ so that $X$ and $Y$ are symmetrically correlated
binary random variables with $\pP(Y \neq X) = p$, (ii) taking
$d(x,\hat{x})$ to be the binary erasure distortion function, (iii)
selecting a value for $D$, and (iv) showing that
$\rho_1(p_{X|Y}p_Y,D)$ is not concave with respect to $p_{Y}$.  By
Lemma~\ref{lem:rho1rho2}, this would imply that $\rho_1(p_{XY},D) \neq
\rho_2(p_{XY},D)$ which, in turn, would imply that
$R_{sum,1}(p_{XY},D)>R_{sum,2}(p_{XY},D)$. In
Theorem~\ref{thm:largeratio} we will show that for certain values of
parameters $p$ and $D$, the two-message sum-rate can be split in such
a way that the ratio $R_1/R_2$ is arbitrarily small and simultaneously
the ratio $R_{sum,1}/(R_1+R_2)$ is arbitrarily large. This will be
proved by explicitly constructing auxiliary variables $V_1,V_2$ and
decoding function $g$ in (\ref{eqn:twomsgrate}). While the explicit
construction of $V_1,V_2$ and $g$ in the proof of
Theorem~\ref{thm:largeratio} may make the implicit proof of
Theorem~\ref{thm:improvement} seem redundant, it is unclear how the
explicit construction can be generalized to other families of source
distributions and distortion functions. The approach followed in the
proof of Theorem~\ref{thm:improvement}, on the other hand, provides an
efficient method to {\em test} whether the best two-message scheme can
strictly outperform the best one-message scheme for {\em more general}
distributed source coding and function computation problems. The
implicit proof naturally points to an explicit construction and was,
in fact, the path taken by the authors to arrive at the explicit
construction.


\section{Main results}\label{sec:mainresult}
\begin{theorem}\label{thm:improvement}
There exists a distortion function $d$, a joint distribution $p_{XY}$,
and a distortion level $D$ for which
$R_{sum,1}(p_{XY},D)>R_{sum,2}(p_{XY},D)$.
\end{theorem}

\begin{proof}
In the light of the discussion in Section~\ref{subsec:ratereduction},
to prove Theorem~\ref{thm:improvement}, it is sufficient to show there
exist $p_{X|Y}$, $d$, and $D$ for which $\rho_1(p_{X|Y}p_Y,D)$ is not
concave with respect to $p_Y$. In particular, it is sufficient to show
that there exist $p_{Y,1}$ and $p_{Y,2}$ such that
\begin{equation}\label{eqn:notconcave}
\rho_1\left(p_{X|Y}\frac{p_{Y,1}+p_{Y,2}}{2},D\right)<
\frac{\rho_1\left(p_{X|Y}p_{Y,1},D\right)+\rho_1\left(p_{X|Y}p_{Y,2},D\right)}{2}.
\end{equation}
Let $\mathcal{X} = \mathcal{Y} = \{0,1\}$, and
$\widehat{\mathcal{X}} = \{0,1,e\}$. Let $d$ be the binary erasure
distortion function, i.e., $d: \{0,1\}\times \{0,e,1\}\rightarrow
\{0,1,\infty\}$ and for $i=0,1$, $d(i,i)=0$, $d(i,1-i)=\infty$, and
$d(i,e)=1$. Let $p_{Y,1}(1) = 1 - p_{Y,1}(0) = p_{Y,2}(0) = 1 -
p_{Y,2}(1) = q$, i.e., $p_{Y,1} = $ Bernoulli$(q)$ and $p_{Y,2} = $
Bernoulli$(\bar q)$.\footnote{For any $a\in [0,1]$, $\bar a:=1-a$. For
the erasure symbol $e$, $\bar e := e$.} Let $p_{X|Y}$ be the
conditional pmf of the binary symmetric channel with crossover
probability $p$, i.e., $p_{X|Y}(1|0) = p_{X|Y}(0|1) = p$. Let
$p_Y:=(p_{Y,1}+p_{Y,2})/2$ which is Bernoulli$(1/2)$. The joint
distribution $p_{XY}=p_Y p_{X|Y}$ is the joint pmf of a pair of doubly
symmetric binary sources (DSBS) with parameter $p$, i.e., if $p_{xy}$
denotes $p_{XY}(x,y)$, then $p_{00} = p_{11} = \bar p/2$ and $p_{01} =
p_{10} = p/2$. For these choices of $p_{X|Y}$, $p_{Y,1}$, $p_{Y,2}$,
$p_Y$, and $d$, we will analyze the left and right sides of
(\ref{eqn:notconcave}) step by step through a sequence of definitions
and propositions and establish the strict inequality for a suitable
choice of $D$. The proofs of all the propositions are given in
Section~\ref{sec:proofs}.

\noindent $\bullet$ \emph{Left-side of (\ref{eqn:notconcave}):} From
(\ref{eqn:onemsgrate}) and (\ref{eqn:rhoi}) we have
\begin{equation}\label{eqn:generalrho1}
\rho_1 (p_{XY},D)= \max_{p_{U|X},g:~ \eE[d(X,g(U,Y))]\leq D} \{
H(X|Y,U)+H(Y|X)\}.
\end{equation}
For the binary erasure distortion and a full support joint source pmf
taking values in binary alphabets, (\ref{eqn:generalrho1}) simplifies
to the expression given in Proposition~\ref{lem:erasure}.
\begin{proposition}\label{lem:erasure} If $\mathcal X=\mathcal Y=\{0,1\}$,
$\supp(p_{XY})=\{0,1\}^2$, $d$ is the binary erasure distortion, and
$D\in \rR$, then $\rho_{1} = \max_{p_{U|X}} (H(X|Y,U) + H(Y|X))$,
where $\mathcal U=\{0,e,1\}$ and
\begin{equation}
p_{U|X}(u|x)= \left\{
\begin{array}{ll}
\alpha_{0e}, & \mbox{ if } x=0, u=e, \\
1-\alpha_{0e}, & \mbox{ if } x=0, u=0, \\
\alpha_{1e}, & \mbox{ if } x=1, u=e, \\
1-\alpha_{1e}, & \mbox{ if } x=1, u=1,\\
0, & \mbox{ otherwise,}
\end{array}\right.\label{eqn:pUX}
\end{equation}
where $\alpha_{0e}, \alpha_{1e}\in [0,1]$ satisfy
$\eE[d(X,U)]=p_X(0)\alpha_{0e}+p_X(1)\alpha_{1e}\leq D$.
\end{proposition}

The expression for $\rho_1$ further simplifies to the one in
Proposition~\ref{prop:erasure} by using $p_{U|X}$ given by
(\ref{eqn:pUX}) in (\ref{eqn:generalrho1}).
\begin{proposition} \label{prop:erasure}
If $\mathcal X=\mathcal Y=\{0,1\}$, $\supp(p_{XY})=\{0,1\}^2$, $d$
is the binary erasure distortion, and $D\in \rR$, then
\begin{equation}
\rho_1(p_{XY},D)=\max_{\scriptstyle \alpha_{0e},\alpha_{1e}\in[0,1]:
\atop \scriptstyle \phi(p_{XY},\alpha_{0e},\alpha_{1e})\leq D}
\psi(p_{XY},\alpha_{0e},\alpha_{1e}), \label{eqn:rho1}
\end{equation}
where
\begin{eqnarray*}
\lefteqn{\psi(p_{XY},\alpha_{0e},\alpha_{1e})}\\&:=&(p_{00}\alpha_{0e}+p_{10}\alpha_{1e})h\left(\frac{p_{00}\alpha_{0e}}{p_{00}\alpha_{0e}+p_{10}\alpha_{1e}}\right)\\
&&+(p_{01}\alpha_{0e}+p_{11}\alpha_{1e})h\left(\frac{p_{01}\alpha_{0e}}{p_{01}\alpha_{0e}+p_{11}\alpha_{1e}}\right)\\
&&+(p_{00}+p_{01})h\left(\frac{p_{00}}{p_{00}+p_{01}}\right)+(p_{11}+p_{10})h\left(\frac{p_{11}}{p_{11}+p_{10}}\right),
\end{eqnarray*}
$\phi(p_{XY},\alpha_{0e},\alpha_{1e}):=p_X(0) \alpha_{0e}+p_X(1)
\alpha_{1e}$, and $h$ is the binary entropy function: $h(\theta) :=
-\theta\log_2\theta -\bar\theta\log_2\bar\theta, \theta \in [0,1]$.
\end{proposition}

Finally, for a DSBS with parameter $p$ and the binary erasure
distortion, $\rho_1$ reduces to the compact expression in
Proposition~\ref{prop:DSBS}.
\begin{proposition}\label{prop:DSBS}
If $d$ is the binary erasure distortion, $D\in [0,1]$, and $p_{XY}$ is
the joint pmf of a DSBS with parameter $p$, then
\begin{equation}
\rho_1(p_{XY},D) =(1+D)h(p).\label{eqn:rho1DSBS}
\end{equation}
\end{proposition}

\noindent $\bullet$ \emph{Right-side of (\ref{eqn:notconcave}):}
Solving the rate reduction functionals in the right-side of
(\ref{eqn:notconcave}) requires solving the maximization problem
(\ref{eqn:rho1}) for asymmetric distributions $p_{X|Y}p_{Y,1}$ and
$p_{X|Y}p_{Y,2}$. Exactly solving this problem is cumbersome but it is
easy to provide a lower bound for the maximum as follows.
\begin{proposition}\label{prop:asymsource}
If $d$ is the binary erasure distortion, $p_{Y,1}$ is Bernoulli$(q)$,
$p_{Y,2}$ is Bernoulli$(\bar q)$, and $p_{X|Y}$ is the conditional pmf
of the binary symmetric channel with crossover probability $p$, then
the inequality
\begin{equation}
\frac{\rho_1(p_{X|Y}p_{Y,1},D)+\rho_1(p_{X|Y}p_{Y,2}, D)}{2}\geq
  C(p,q,\alpha_{0e},1)\label{eqn:rho1andC}
\end{equation}
holds for $D=\eta(p,q,\alpha_{0e},1)$, where
\begin{eqnarray}
C(p,q,\alpha_{0e},\alpha_{1e})&:=&\psi(p_{X|Y}p_{Y,1},\alpha_{0e},\alpha_{1e}),\nonumber\\
\eta(p,q,\alpha_{0e},\alpha_{1e})&:=&\phi(p_{X|Y}p_{Y,1},\alpha_{0e},\alpha_{1e}).
\nonumber
\end{eqnarray}
\end{proposition}
\begin{remark}\label{remark:achieveC} The rate-distortion tuple
$(H(X|Y)+H(Y|X)-C(p,q,\alpha_{0e},1),\eta(p,q,\alpha_{0e},1))$ is
admissible for one-message source coding for joint source distribution
$p_{X|Y}p_{Y,1}$ and corresponds to choosing $p_{U|X}$ given by
(\ref{eqn:pUX}) with $\alpha_{1e}=1$ and the decoding function
$g(u,y)=u$. Since this choice of $p_{U|X}$ and $g$ may be suboptimal,
$C(p,q,\alpha_{0e},1)$ is only a lower bound for the rate reduction
functional.
\end{remark}

\noindent $\bullet$ \emph{Comparing left and right sides of
(\ref{eqn:notconcave}):} The left-side of (\ref{eqn:notconcave}) and
the lower bound of the right-side of (\ref{eqn:notconcave}) can be
compared as follows.
\begin{proposition}\label{prop:takinglimit}
Let $d$ be the binary erasure distortion, $p_Y$ be Bernoulli$(1/2)$,
and $p_{X|Y}$ be the binary symmetric channel with parameter $p$.  For
all $q\in(0,1/2)$ and all $\alpha_{0e}\in (0,1)$, there exists $
p\in(0,1)$ such that the strict inequality $\rho_1(p_{XY},D) <
C(p,q,\alpha_{0e},1)$ holds for $D=\eta(p,q,\alpha_{0e},1)$.
\end{proposition}

Since the left-side of (\ref{eqn:notconcave}) is strictly less than a
lower bound of the right-side of (\ref{eqn:notconcave}), the strict
inequality (\ref{eqn:notconcave}) holds, which completes the proof of
Theorem~\ref{thm:improvement}.
\end{proof}

Theorem~\ref{thm:largeratio} quantifies the multiplicative reduction
in the sum-rate that is possible with two messages.
\begin{theorem}\label{thm:largeratio}
If $d$ is the binary erasure distortion and $p_{XY}$ the joint pmf of
a DSBS with parameter $p$, then for all $L >0$ there exists an
admissible two-message rate-distortion tuple $(R_1,R_2, D)$ such that
$R_{sum,1}(p_{XY},D)/(R_1+R_2)>L$ and $R_1/R_2<1/L$.
\end{theorem}

\begin{proof}
We will explicitly construct $p_{V_1|Y},p_{V_2|XV_1}$, and $g$ in
(\ref{eqn:rateregion}) which lead to an admissible tuple
$(R_1,R_2,D)$. Let $p_{V_1|Y}$ be the conditional pmf of the binary
symmetric channel with crossover probability $q$.  Let the conditional
distribution $p_{V_2|XV_1}(v_2|x,v_1)$ have the form described in
Table~\ref{tab:pV2} and let $g(v_1,v_2,y):=v_2$.
\begin{table}[!htb]
\vglue -0.3cm
  \centering
 \caption{Conditional distribution $p_{V_2|XV_1}$}
\vglue -0.3cm
  \begin{tabular}{|c|c|c|c|}
  \hline
   $p_{V_2|XV_1}$ & $v_2=0$ & $v_2=e$& $v_2=1$ \\
   \hline
  $x=0,v_1=0$ & $1-\alpha$ & $\alpha$ & $0$ \\
  \hline
  $x=1,v_1=0$ & $0$ & $1$ & $0$\\
  \hline
  $x=0,v_1=1$ & $0$ & $1$ & $0$ \\
  \hline
  $x=1,v_1=1$ & $0$ & $\alpha$ & $1-\alpha$\\
  \hline
\end{tabular}
\label{tab:pV2}
\vglue -0.3cm
\end{table}

The corresponding rate-distortion tuple can be shown to satisfy the
following property.
\begin{proposition}\label{prop:largeratio}
Let $d$ be the binary erasure distortion and let $p_{XY}$ be the joint
pmf of a DSBS with parameter $p$. For $p_{V_1|Y}, p_{V_2|XV_1}$, and
$g$ as described above, and all $L >0$, there exist parameters
$p,q,\alpha$ such that the two-message rate-distortion tuple
$(R_1,R_2, D)$ given by $R_1 = I(Y;V_1|X)$, $R_2 = I(X;V_2|Y,V_1)$, $D
= \eE[d(X,V_2)]$ satisfies $R_{sum,1}(p_{XY},D)/(R_1+R_2)>L$ and
$R_1/R_2<1/L$.
\end{proposition}


This completes the proof of Theorem~\ref{thm:largeratio}. \end{proof}

The conditional pmfs $p_{V_1|Y}$ and $p_{V_2|XV_1}$ in the proof of
Theorem~\ref{thm:largeratio} are related to the conditional pmf
$p_{U|X}$ in the proof of Theorem~\ref{thm:improvement} as follows.
Given $V_1=0$, the conditional distribution
$p_{XYV_2|V_1}(x,y,v_2|0)=p_{Y,1}(y) p_{X|Y}(x|y)p_{U|X}(v_2|x)$,
where $p_{U|X}$ is given by (\ref{eqn:pUX}) with $\alpha_{0e}=\alpha$
and $\alpha_{1e}=1$. Given $V_1=1$, the conditional distribution
$p_{XYV_2|V_1}(x,y,v_2|1)=p_{Y,2}(y) p_{X|Y}(x|y)p_{U|X}(v_2|x)$,
where $p_{U|X}$ is given by (\ref{eqn:pUX}) with $\alpha_{1e}=\alpha$
and $\alpha_{0e}=1$. Conditioning on $V_1$, in effect, decomposes the
two-message problem into two one-message problems that were analyzed
in the proof of Theorem~\ref{thm:improvement}.

\section{Proofs}\label{sec:proofs}
\noindent{\em Proof of Proposition~\ref{lem:erasure}:}
Given a general $p_{U|X}$ and $g$ satisfying the original constraint
in (\ref{eqn:generalrho1}), we will construct $U^*$ satisfying the
stronger constraints in Proposition~\ref{lem:erasure} with an
objective function that is not less than the original one as follows.

Without loss of generality, we assume $\supp(p_U)=\mathcal U$. For
$i=0,1$, let $\mathcal U_i:=\{u \in \mathcal U: p_{X|U}(i|u)=1\}$.
Let $\mathcal U_e:=\{u\in \mathcal U: p_{X|U}(1|u)\in(0,1)\}$. Then
$\{\mathcal U_1, \mathcal U_0, \mathcal U_e\}$ forms a partition of
$\mathcal U$. For each $u\in \mathcal U_e$, since $p_{XY|U}(x,y|u)>0$
for all $(x,y)\in \{0,1\}^2$, it follows that $g(u,y=0)=g(u,y=1)=e$
must hold, because otherwise $\eE(d(X,g(U,Y)))=\infty$. But for every
$u\in \mathcal U_i$, $i=0,1$, $g(u,y)$ may equal $i$ or $e$ but not
$(1-i)$ to get a finite distortion. When we replace $g$ by
\begin{equation*}
g^*(u,y)= \left\{
\begin{array}{ll}
i, & \mbox{ if } u\in \mathcal U_i, i=0,1, \\
e, & \mbox{ if } u\in \mathcal U_e,
\end{array}\right.
\end{equation*}
the distortion for $u\in \mathcal U_i, i=0,1,$ is reduced to zero, and
the distortion for $u\in \mathcal U_e$ remains unchanged.  Therefore
we have $\eE(d(X,g^*(U,Y)))\leq \eE(d(X,g(U,Y)))\leq D$.  Note that
$g^*(U,Y)$ is completely determined by $U$. Let $U^* :=g^*(U,Y)$. Then
$U^*=i$ iff $U\in \mathcal U_i, i=\{0,1,e\}$. The objective function
$H(X|Y,U)+H(Y|X) = H(X|Y,U,U^*)+H(Y|X) \leq H(X|Y,U^*)+H(Y|X)$, which
completes the proof. \hspace*{\fill}~\QED\\

\noindent{\em Proof of Proposition~\ref{prop:DSBS}:}

For a fixed $p_{XY}$, $H(X|Y,U)+H(Y|X)$ is concave with respect to
$p_{XYU}$ and therefore also $p_{U|X}$. Since $p_{U|X}$ is linear with
respect to $(\alpha_{0e},\alpha_{1e})$,
$\psi(p_{XY},\alpha_{0e},\alpha_{1e})=H(X|Y,U)+H(Y|X)$ is concave with
respect to $(\alpha_{0e},\alpha_{1e})$.

The maximum in (\ref{eqn:rho1}) can be achieved along the axis of
symmetry given by $\alpha_{1e}=\alpha_{0e}$ because (i) $\psi$ and
$\phi$ are both symmetric with respect to $\alpha_{0e}$ and
$\alpha_{1e}$, i.e.,
$\psi(p_{XY},\alpha_{0e},\alpha_{1e})=\psi(p_{XY},\alpha_{1e},\alpha_{0e})$
and
$\phi(p_{XY},\alpha_{0e},\alpha_{1e})=\phi(p_{XY},\alpha_{1e},\alpha_{0e})$,
and (ii) $\psi(p_{XY},\alpha_{0e},\alpha_{1e})$ is a concave function
of $(\alpha_{0e},\alpha_{1e})$. When $D\in [0,1]$, $\rho_1$ can be
further simplified as follows.
\begin{equation*}
\rho_1(p_{XY},D) =\max_{\scriptstyle \alpha_{0e}=\alpha_{1e}\in [0,D]}
\psi(p_{XY},\alpha_{0e},\alpha_{1e})=(1+D)h(p),
\end{equation*}
which completes the proof. \hspace*{\fill}~\QED\\

\noindent{\em Proof of Proposition~\ref{prop:asymsource}:}
\begin{table}[!htb]
\vglue -0.3cm
  \centering
 \caption{Joint distribution $p_{X|Y}p_{Y,1}$}
\vglue -0.3cm
  \begin{tabular}{|c|c|c|}
  \hline
   $p_{X|Y}p_{Y,1}$ & $y=0$ & $y=1$ \\
   \hline
  $x=0$ & $\bar p\bar q$ & $pq$ \\
  \hline
  $x=1$ & $p\bar q$ & $\bar p q$ \\
  \hline
\end{tabular}
\label{tab:pXY1}
\vglue -0.3cm
\end{table}
For the joint pmf $p_{X|Y}p_{Y,1}$ summarized in Table~\ref{tab:pXY1},
functions $\psi$ and $\eta$ simplify even further to special functions
of $(p,q,\alpha_{0e},\alpha_{1e})$ as follows:
\begin{eqnarray}
C(p,q,\alpha_{0e},\alpha_{1e})&=&\psi(p_{X|Y}p_{Y,1},\alpha_{0e},\alpha_{1e})
\nonumber\\
&=&\bar q(\bar p  \alpha_{0e}+p \alpha_{1e})h\left(\frac{\bar p \alpha_{0e}}{\bar p \alpha_{0e}+ p
\alpha_{1e}}\right)\nonumber\\
&&+q(p \alpha_{0e}+\bar p
\alpha_{1e})h\left(\frac{p\alpha_{0e}}{p\alpha_{0e}+\bar p
  \alpha_{1e}}\right)\nonumber\\
&&+(\bar p \bar q + p q) h\left(\frac{\bar p \bar q}{\bar p \bar q +
p q}\right)\nonumber\\
&&+(\bar p q+p \bar q)h\left(\frac{\bar p q}{\bar p q+p \bar
q}\right),\label{eqn:Cpq}\\
\eta(p,q,\alpha_{0e},\alpha_{1e})&=&\phi(p_{X|Y}p_{Y,1},\alpha_{0e},\alpha_{1e})
\nonumber\\
&=&(\bar p \bar q + p q)\alpha_{0e}+(\bar p q+p \bar q)\alpha_{1e}.\nonumber
\end{eqnarray}
Observe that $C(p,q,\alpha_{0e},\alpha_{1e})=C(p,\bar
q,\alpha_{1e},\alpha_{0e})$, and $\eta(p,q,\alpha_{0e},\alpha_{1e}) =
\eta(p,\bar q,\alpha_{1e},\alpha_{0e})$ hold. Therefore we have
\begin{eqnarray*}
\rho_1(p_{X|Y}p_{Y,2},D)&=&\max_{\scriptstyle
\alpha_{0e},\alpha_{1e}\in[0,1]: \atop \scriptstyle \eta(p,\bar q,
\alpha_{0e},\alpha_{1e})\leq D} C(p,\bar
q,\alpha_{0e},\alpha_{1e})\\
&=&\max_{\scriptstyle \alpha_{0e},\alpha_{1e}\in[0,1]: \atop
\scriptstyle
\eta(p,q, \alpha_{1e},\alpha_{0e})\leq D} C(p,q,\alpha_{1e},\alpha_{0e})\\
&=&\rho_1(p_{X|Y}p_{Y,1},D).
\end{eqnarray*}
It follows that
\begin{eqnarray}
\frac{\rho_1(p_{X|Y}p_{Y,1},D)+\rho_1(p_{X|Y}p_{Y,2},
  D)}{2}&=&\rho_1(p_{X|Y}p_{Y,1},D)\nonumber\\
  &\geq& C(p,q,\alpha_{0e},1)\nonumber
\end{eqnarray}
holds for $D=\eta(p,q,\alpha_{0e},1)$. \hspace*{\fill}~\QED

\noindent{\em Proof of Proposition~\ref{prop:takinglimit}:}

Since $D=\eta(p,q,\alpha_{0e},1)\in[0,1]$ always holds, we have
$\rho_1(p_{XY},D)=(1+D)h(p)$ due to (\ref{eqn:rho1DSBS}). We will
show that for any fixed $q\in(0,1/2)$ and $\alpha_{0e}\in (0,1)$,
$\lim_{p\rightarrow 0} C(p,q,\alpha_{0e},1)/h(p)> \lim_{p\rightarrow
0} (1+D)$ holds, which implies that $\exists p\in(0,1)$ such that
$C(p,q,\alpha_{0e},1)/h(p)>(1+D)$, which, in turn, implies
Proposition~\ref{prop:takinglimit}. It is convenient to use the
following lemma to analyze the limits.
\begin{lemma}\label{lem:entropyratio}
Let $f(p)$ be a function differentiable around $p=0$ such that
$f(0)=0$ and $f'(0)>0$. Then
\[\lim_{p\rightarrow 0} \frac{h(f(p))}{h(p)}=f'(0)\]
\end{lemma}
\begin{proof}
Applying the l'H\^{o}pital rule several times, we have
\begin{eqnarray*}
\lim_{p\rightarrow 0} \frac{h(f(p))}{h(p)}&=&\lim_{p\rightarrow 0}
\frac{\ln(1-f(p))-\ln f(p)}{\ln(1-p)-\ln p } f'(0)\\
&=&\lim_{p\rightarrow 0} \frac{\ln f(p)}{\ln p}
f'(0)\\&=&\lim_{p\rightarrow 0} \frac{p}{f(p)} (f'(0))^2\\&=&f'(0),
\end{eqnarray*}
which completes the proof of Lemma~\ref{lem:entropyratio}.
\end{proof}

Applying Lemma~\ref{lem:entropyratio}, we have
\begin{eqnarray}
 &&\lim_{p\rightarrow 0}
\frac{C(p,q,\alpha_{0e},1)}{h(p)} =2 - q(1-\alpha_{0e}),\label{eqn:limitC}\\
&&\lim_{p\rightarrow 0} (1+D) =
2-\bar q (1- \alpha_{0e}),\label{eqn:limitrho1}\\
&&\lim_{p\rightarrow 0}
\left(\frac{C(p,q,\alpha_{0e},1)}{h(p)}-(1+D)\right) =
(1-2q)(1-\alpha_{0e}).\nonumber
\end{eqnarray}
Therefore for any $\alpha_{0e}\in(0,1)$ and $q\in(0,1/2)$, there
exists a small enough $p$ such that $C(p,q,\alpha_{0e},1)>(1+D)$
holds, which completes the proof.\hspace*{\fill}~\QED\\

%

\noindent{\em Proof of Proposition~\ref{prop:largeratio}:}

For the rate-distortion tuple $(R_1,R_2,D)$ corresponding to the
choice of $p_{V_1|Y}$, $p_{V_2|XV_1}$ and $g$ described in the proof
of Theorem~\ref{thm:largeratio}, we have (i) $R_1 =
I(Y;V_1|X)=H(Y|X)- C_2(p,q)$, where $C_2(p,q)$ is the sum of the
last two terms in (\ref{eqn:Cpq}); (ii)
$R_2=I(X;V_2|Y,V_1)=2h(p)-C(p,q,\alpha,1)-R_1$; and (iii)
$D=\eta(p,q,\alpha,1)$. It follows that
\begin{eqnarray*}
\lim_{p\rightarrow 0} \frac{R_{1}}{h(p)} \!\!\!&=&\!\!\! 0,\\
\lim_{p\rightarrow 0} \frac{R_{2}}{h(p)} \!\!\!&=&\!\!\!
2-\lim_{p\rightarrow 0}
\frac{C(p,q,\alpha,1)}{h(p)}-\lim_{p\rightarrow 0} \frac{R_1}{h(p)}
=q(1-\alpha).
\end{eqnarray*}
Therefore for all $q>0$ and $\alpha\in (0,1)$, we have
\begin{equation}\lim_{p\rightarrow
0}\frac{R_1}{R_2}=0.\label{eqn:R1R2}
\end{equation}
For the one-message rate-distortion function, we have
$R_{sum,1}(p_{XY},D)=2h(p)-\rho_1(p_{XY},D)$, where
$\rho_1(p_{XY},D)$ is given by (\ref{eqn:rho1DSBS}). Therefore we
have
\[
\lim_{p\rightarrow 0} \frac{R_{sum,1}(p_{XY},D)}{h(p)} {=}
2-\lim_{p\rightarrow 0} \frac{\rho_1(p_{XY},D)}{h(p)} {=} \bar q
(1-\alpha),
\]
which implies that
\begin{equation}
\lim_{p\rightarrow 0} \frac{R_{sum,1}(p_{XY},D)}{R_1+R_2} = \frac{\bar
q}{q}. \label{eqn:Rsum1R1plusR2}
\end{equation}
For any $L >0$, we can always find a small enough $q>0$ such that
$\bar q/q > L+1$. Due to (\ref{eqn:R1R2}) and
(\ref{eqn:Rsum1R1plusR2}), there exists $p>0$ such that $R_1/R_2<1/L$
and $R_{sum,1}/(R_1+R_2)>L$.\hspace*{\fill}~\QED

\begin{remark}
The convergence of the limit analyzed in Lemma~\ref{lem:entropyratio}
is actually slow, because the logarithm function increases to infinity
slowly. The consequence is that if one chooses a small $q$ to get
$R_{sum,1}/(R_1+R_2)$ close to the limit $\bar q/q$, then $p$ needs to
be very small. For example, when $q=1/10, \alpha_{0e}=1/2$, $\bar
q/q=9$, with $p=10^{-200}$, we get $R_{sum,1}/R_{sum,2}^*\approx
8.16$. This, however, does not mean that the benefit of multiple
messages only occurs in extreme cases. In numerical computations we
have observed that for the erasure distortion, the gain for certain
asymmetric sources can be much more than that for the DSBS example
analyzed in this paper. The DSBS example was chosen in this paper only
because it is easy to analyze.
\end{remark}

\section*{Acknowledgment} The second author would like to thank
Dr.~Vinod Prabhakaran for introducing him to the unresolved question
in \cite{Kaspi1985}.


\appendices
\renewcommand{\theequation}{\thesection.\arabic{equation}}
\setcounter{equation}{0}

\footnotesize

\bibliography{newbibfile}

\begin{thebibliography}{1}
\providecommand{\url}[1]{#1}
\csname url@samestyle\endcsname
\providecommand{\newblock}{\relax}
\providecommand{\bibinfo}[2]{#2}
\providecommand{\BIBentrySTDinterwordspacing}{\spaceskip=0pt\relax}
\providecommand{\BIBentryALTinterwordstretchfactor}{4}
\providecommand{\BIBentryALTinterwordspacing}{\spaceskip=\fontdimen2\font plus
\BIBentryALTinterwordstretchfactor\fontdimen3\font minus
  \fontdimen4\font\relax}
\providecommand{\BIBforeignlanguage}[2]{{%
\expandafter\ifx\csname l@#1\endcsname\relax
\typeout{** WARNING: IEEEtran.bst: No hyphenation pattern has been}%
\typeout{** loaded for the language `#1'. Using the pattern for}%
\typeout{** the default language instead.}%
\else
\language=\csname l@#1\endcsname
\fi
#2}}
\providecommand{\BIBdecl}{\relax}
\BIBdecl

\bibitem{Kaspi1985}
A.~H. Kaspi, ``Two-way source coding with a fidelity criterion,'' \emph{IEEE
  Trans.~Inf.~Theory}, vol.~31, no.~6, pp. 735--740, Nov. 1985.

\bibitem{Dakehe}
E.~Yang and D.~He, ``On interactive encoding and decoding for lossless source
  coding with decoder only side information,'' in \emph{Proc.~IEEE
  Int.~Symp.~Information~Theory}, {Toronto, Canada}, {{Jul.~6--11,}} 2008, pp.
  419--423.

\bibitem{OrlitskyRoche}
A.~Orlitsky and J.~R. Roche, ``Coding for computing,'' \emph{IEEE
  Trans.~Inf.~Theory}, vol.~47, no.~3, pp. 903--917, Mar. 2001.

\bibitem{ISIT08}
{N.~Ma and P.~Ishwar}, ``{Two-terminal distributed source coding with
  alternating messages for function computation},'' in \emph{Proc.~IEEE
  Int.~Symp.~Information~Theory}, {Toronto, Canada}, {{Jul.~6--11,}} 2008, pp.
  51--55, arXiv:0801.0756v4.

\bibitem{WynerZiv1976}
A.~Wyner and J.~Ziv, ``The rate-distortion function for source coding with side
  information at the decoder,'' \emph{IEEE Trans.~Inf.~Theory}, vol.~22, no.~1,
  pp. 1--11, Jan. 1976.

\bibitem{Allerton09}
N.~Ma and P.~Ishwar, ``Infinite-message distributed source coding for
  two-terminal interactive computing,'' in \emph{Proc.~47th Annu. Allerton
  Conf. Communication, Control, Computing}, {Monticello, IL},
  {{Sep.~30--Oct.~7,}} 2009.

\bibitem{Allerton09arXiv}
------, ``Infinite-message distributed source coding for two-terminal
  interactive computing,'' \emph{arXiv:0908.3512v2}.

\end{thebibliography}
\end{document}